\documentclass[conference,10pt]{IEEEtran}

\usepackage[left=0.63in,right=0.625in,top=0.75in,bottom=1in]{geometry}

\usepackage[T1]{fontenc}

\usepackage{cite}
\usepackage{graphicx}
\graphicspath{{Figure/}}
\usepackage[cmex10]{amsmath}
\usepackage{amsthm}
\usepackage{amssymb}
\usepackage{algorithmic}
\usepackage{array}
\usepackage{color}
\usepackage{epstopdf}
\usepackage{amsfonts}
\usepackage{siunitx}
\usepackage{soul}
\usepackage{xurl}
\usepackage[acronym,shortcuts]{glossaries}

\usepackage{pgfplots}
\pgfplotsset{compat=1.3}
\usepackage{tikz}							
\usetikzlibrary{shapes}
\usetikzlibrary{spy}
\usetikzlibrary{circuits}
\usetikzlibrary{arrows}	
\usepackage{subfig}

\newcommand{\vect}[1]{\boldsymbol{\mathrm{#1}}}
\newcommand{\mat}[1]{\boldsymbol{\mathrm{#1}}}

\usepackage{mathtools}

\theoremstyle{definition}
\newtheorem{theorem}{Theorem}
\newtheorem{proposition}{Proposition}
\newtheorem{lemma}{Lemma}
\newtheorem{corollary}{Corollary}

\makeglossaries
\newacronym[plural=BSs,firstplural=base stations (BSs)]{bs}{BS}{base station}
\newacronym[plural=PAs,firstplural=power amplifiers (PAs)]{pa}{PA}{power amplifier}
\newacronym{sdr}{SDR}{signal-to-distortion ratio}
\newacronym{snr}{SNR}{signal-to-noise ratio}
\newacronym{sndr}{SNDR}{signal-to-noise-and-distortion ratio}
\newacronym{los}{LoS}{line-of-sight}
\newacronym{z3ro}{Z3RO}{zero third-order distortion}
\newacronym{mimo}{MIMO}{multiple input multiple output}
\newacronym{mrt}{MRT}{maximum ratio transmission}
\newacronym{dpd}{DPD}{digital pre-distortion}
\newacronym{ota}{OTA}{over-the-air}

\newacronym{ula}{ULA}{uniform linear array}

\newacronym{psd}{PSD}{power spectral density}
\newacronym{aclr}{ACLR}{adjacent channel leakage ratio}
\newacronym{cpu}{CPU}{central processing unit}
\newacronym{iid}{i.i.d.}{independently and identically distributed}
\newacronym{ue}{UE}{user equipment}
\newacronym{nlos}{NLoS}{non-line-of-sight}
\newacronym{hpbm}{HPBM}{half power beam width}
\newacronym{rts}{RTS}{ray tracing simulator}
\newacronym{ag}{AG}{array gain}
\newacronym{fft}{FFT}{fast Fourier transform}
\newacronym{ib}{IB}{in-band}
\newacronym{im}{IM}{intermodulation}
\newacronym{mf}{MF}{matched filtering}
\newacronym{mr}{MR}{maximum ratio}
\newacronym{ofdm}{OFDM}{orthogonal frequency division multiplexing}
\newacronym{oob}{OOB}{out-of-band}
\newacronym{rx}{RX}{Receiver}
\newacronym{tx}{TX}{transmitter}
\newacronym{zf}{ZF}{zero forcing}
\newacronym{dl}{DL}{downlink}

\newacronym{ls}{LS}{least squares}
\newacronym{ml}{ML}{maximum likelihood}
\newacronym{lmmse}{LMMSE}{linear minimum mean squared error}
\newacronym{mse}{MSE}{mean squared error}

\begin{document}

\title{Optimal Training Design for Over-the-Air Polynomial Power Amplifier Model Estimation\vspace{-0.5em}}

\author{\IEEEauthorblockN{\IEEEauthorrefmark{1}Fran\c{c}ois Rottenberg, \IEEEauthorrefmark{1}Thomas Feys, \IEEEauthorrefmark{2}Nuutti Tervo
	}
	\IEEEauthorblockA{\IEEEauthorrefmark{1}KU Leuven, ESAT-WaveCore, Ghent Technology Campus, 9000 Ghent, Belgium
	}
	\IEEEauthorblockA{\IEEEauthorrefmark{2}Centre for Wireless Communications - Radio Technologies (CWC-RT), University of Oulu, Oulu FI-90014, Finland\vspace{-1.5em}
	}
}

\maketitle

\begin{abstract}
	The current evolution towards a massive number of antennas and a large variety of transceiver architectures forces to revisit the conventional techniques used to improve the fundamental \gls{pa} linearity-efficiency trade-off. Most of the digital linearization techniques rely on \gls{pa} measurements using a dedicated feedback receiver. However, in modern systems with large amount of RF chains and high carrier frequency, dedicated receiver per RF chain is costly and complex to implement. This issue can be addressed by measuring PAs over the air, but in that case, this extra signalling is sharing resources with the actual data transmission. In this paper, we look at the problem from an estimation theory point of view so as to minimize pilot overhead while optimizing estimation performance. We show that conventional results in the mathematical statistics community can be used. We find the \gls{ls} optimal training design, minimizing the maximal \gls{mse} of the reconstructed \gls{pa} response over its whole input range. As compared to uniform training, simulations demonstrate a factor 10 reduction of the maximal \gls{mse} for a $L=7$ PA polynomial order. Using prior information, the LMMSE estimator can achieve an additional gain of a factor up to 300 at low \gls{snr}.
\end{abstract}

\begin{IEEEkeywords}
	Optimal training, power amplifier, calibration.
\end{IEEEkeywords}

\section{Introduction}\label{section:Introduction}

\glsresetall
\bstctlcite{IEEEexample:BSTcontrol}


The \glspl{pa} account for a large part of energy consumption of base stations~\cite{auer11}. Their operating regime faces a fundamental trade-off between linearity and efficiency~\cite{lavr10}. While linearity is desirable to avoid signal degradation and out-of-band emission, driving the \gls{pa} closer to saturation improves its efficiency while the signal suffers from more nonlinearity
. This is particularly challenging in latest broadband technologies where the combination of multicarrier transmission and precoding over to multiple users, leads to a high peak-to-average power ratio, requiring a large linear range of the \gls{pa}~\cite{Fager2019}.

\subsection{State-of-the-Art}

A widely used solution to improve the linearity-efficiency trade-off is the use of \gls{dpd} before the \gls{pa} to linearize its response and allow to drive it closer to saturation~\cite{cripps2006rf}. A recently proposed solution is the use of distortion-aware precoding such as the Z3RO precoder~\cite{Z3RO_journal,10279810}. These techniques require the \gls{pa} response to be properly estimated. A conventional approach is to use a feedback loop after the \gls{pa} within the transmitter chain~\cite{huang_signal_2014}. Unfortunately, in massive MIMO and/or mm-wave systems, this is not always possible or desirable. Replicating such a feedback mechanism per RF chain is expensive in terms of hardware. 
An interesting alternative is to train the \gls{dpd} using the \gls{ota} wireless link with either near-field probes~\cite{9780639}, far-field observation antennas or the entire radio link~\cite{8877285,tervo_thesis}. As this training signalling uses the wireless channel, it is important to minimize the signalling overhead to avoid penalizing data transmission.


To the best of our knowledge, despite this rich literature on array linearization, there is a lack of an estimation-theoretical approach of optimal training design to estimate \gls{pa} model parameters. Moreover, when using a dedicated feedback loop directly after the \gls{pa}, observations are typically obtained continuously at a very high \gls{snr} and noise is generally neglected~\cite{huang_signal_2014}. In contrast, \gls{ota} training suffers from channel attenuation, multipath propagation, and limited \gls{snr} due to the lower received signal strength. Hence, the training design should be optimized to maximize performance while minimizing the signalling overhead.

\subsection{Contributions}

In this paper, we formalize the noisy \gls{ota} estimation problem and rederive the \gls{ls} and \gls{lmmse} estimators for a nonlinear memoryless polynomial \gls{pa} model. We derive the LS optimal training in the sense of minimizing the generalized variance (determinant) of the error covariance matrix. To do this, we show that the problem can be viewed as a D-optimal design, well known in mathematical statistics. Interestingly, the optimal design also minimizes the maximal reconstruction \gls{mse} when predicting the \gls{pa} response over the entire PA input range. The optimal training design can also be used for the LMMSE estimator but is only optimal at high SNR. Simulation results demonstrate the superiority of the optimal design for the LS estimator. We also show, that if enough prior information is available on the PA response, the LMMSE estimator can provide a significant gain at low SNR, especially if the phase of the channel is known (coherent case).

{\textbf{Notations}}: Vectors and matrices are by bold lowercase and uppercase letters $\vect{a}$ and $\mat{A}$. Superscripts $^*$, $^T$ and $^H$ stand for conjugate, transpose and Hermitian transpose. The symbols $\mathbb{E}(.)$, $\angle(.)$, $\|\mat{A}\|$ and $|\mat{A}|$ denote the expectation, phase, Frobenius norm and determinant. $\mathcal{N}(\mu,\sigma^2)$ denotes a normal distribution with mean $\mu$ and variance $\sigma^2$. 

\section{Power Amplifier Model Estimation}\label{section:transmission_model}

Without loss of generality, we consider a single \gls{pa} related to a transmit antenna, which response needs to be estimated. To do this, over-the-air calibration is considered based on an observational receiver and a feedback loop. The framework can be straightforwardly extended to multiple transmit antennas by using orthogonal transmit sequences and decoupling the problem per PA. 
A series of $N$ pilots are sent, denoted by $s_n,\ n=0,...,N-1$. The channel is assumed narrowband and constant during training. The received signal is then
\begin{align*}
	r_n=h f(s_n) + w_n
\end{align*}
where $w_n$ is additive noise and $h$ is the complex channel coefficient. The function $f()$ is the PA transfer function to be estimated. It is here considered to follow a nonlinear quasi-memoryless polynomial model 
\begin{align*}
	f(s)&=\sum_{l=1}^{L}\beta_{l}s|s|^{l-1}
\end{align*}
where $L$ is the model order and $\beta_l \in \mathbb{C}$ are the polynomial coefficients to be estimated. One can note the presence of even-order nonlinear terms, which are often neglected when considering bandpass signals. Still, motivated by~\cite{ding_effects_2004}, we include them to improve both modelling accuracy while using low order polynomials, with better numerical properties and also making the following optimization problem more standard. Without further assumptions, the problem has $L+1$ complex unknowns including the $L$ beta coefficients plus the channel coefficient $h$. In practice, channel estimation will be used at the receiver so that the equivalent linear channel gain $h\beta_1$ will be estimated. We are interested in the relation of higher order terms with respect to this one. To solve this ambiguity, we set $h=1$, which we will rediscuss in~Section~\ref{section:simulation_results}. 
We thus have
\begin{align*}
	r_n&=\sum_{l=1}^{L}\beta_{l}s_n|s_n|^{l-1}+w_n.
\end{align*}
In vector form, this gives a linear observation model
\begin{align*}
	\vect{r}&=\mat{\Phi} \vect{\beta} + \vect{w}
\end{align*}
where $\vect{r}=(r_0,...,r_{N-1})^T$, $\vect{\beta}=(\beta_1,...,\beta_{L})^T \in \mathbb{C}^{L\times 1}$, $\vect{w}=(w_0,...,w_{N-1})^T$ and
\begin{align*}
	\mat{\Phi}&=\begin{pmatrix}
		s_0 & \hdots & s_{0}|s_0|^{L-1}\\
		\vdots & \ddots & \vdots \\
		s_{N-1} & \hdots & s_{N-1}|s_{N-1}|^{L-1}
	\end{pmatrix} \in \mathbb{C}^{N\times L}.
\end{align*}
In the literature, other polynomial bases have been considered, e.g., orthogonal bases for given distributions of the input signals~\cite{raich_orthogonal_2004b}. Generally considering another basis $\mat{\Psi} \in \mathbb{C}^{N\times L}$, we can write $\mat{\Psi}=\mat{\Phi}\mat{U}$ where $\mat{U}\in\mathbb{C}^{L \times L}$ is a full rank matrix. We would then have $\vect{r}=\mat{\Psi}\vect{\alpha}+\vect{w}$ where $\vect{\alpha}=\mat{U}\vect{\beta}$.

\subsection{Least Squares Estimator}

The classical \gls{ls} estimate is given by~\cite{kay_fundamentals_1993,huang_signal_2014}
\begin{align*}
	\hat{\vect{\beta}}_{\mathrm{LS}}=\arg \min_{\vect{\beta}} \|\vect{r}-\mat{\Phi}\vect{\beta}\|^2=(\mat{\Phi}^H\mat{\Phi})^{-1}\mat{\Phi}^H\vect{r}.
\end{align*}
We here consider the noise samples $w_n$ identically and independently complex circularly symmetric distributed with zero mean and variance $\sigma^2$ so that the LS estimate corresponds to the \gls{ml} estimate~\cite{kay_fundamentals_1993}. The error covariance matrix is
\begin{align*}
	\mat{C}_{\mathrm{LS}}&=\sigma^2(\mat{\Phi}^H\mat{\Phi})^{-1}.
\end{align*}
We can use the \gls{ls} estimate to predict, \textit{i.e.}, reconstruct, the PA response at $\tilde{N}$ values $\tilde{s}_n,n=0,...,\tilde{N}-1$ as $\tilde{\mat{\Phi}}\hat{\vect{\beta}}_{\mathrm{LS}}$ with 
\begin{align*}
	\tilde{\mat{\Phi}}&=\begin{pmatrix}
		\tilde{s}_0 & \hdots & \tilde{s}_{0}|\tilde{s}_0|^{L-1}\\
		\vdots & \ddots & \vdots \\
		\tilde{s}_{\tilde{N}-1} & \hdots & \tilde{s}_{\tilde{N}-1}|\tilde{s}_{\tilde{N}-1}|^{L-1}
	\end{pmatrix} \in \mathbb{C}^{\tilde{N}\times L}.
\end{align*}
The resulting estimate is unbiased with error covariance matrix
\begin{align}
	\tilde{\mat{C}}_{\mathrm{LS}}&=\sigma^2\tilde{\mat{\Phi}}(\mat{\Phi}^H\mat{\Phi})^{-1}\tilde{\mat{\Phi}}^H. \label{eq:predicted_C_LS}
\end{align}

\subsection{Linear Minimum Mean Squared Estimator}

How to choose the polynomial order $L$? Intuitively, increasing it would only help as it would allow a better fit. Unfortunately, this is not the case for the conventional \gls{ls} estimator due to three reasons. i) It may suffer from numerical instability if a high polynomial order $L$ is considered~\cite{huang_signal_2014,raich_orthogonal_2004b}. In other words, the matrix inverse $(\mat{\Phi}^H\mat{\Phi})^{-1}$ becomes ill conditioned. ii) As will be shown in Corol.~\ref{corol:minimax}, the prediction \gls{mse} grows with $L$. Intuitively, more coefficients need to be estimated leading to less noise averaging. iii) Moreover, the minimum number of pilots $N$ has to grow with $L$, since the matrix $\mat{\Phi}$ needs at least $L$ pilots with different magnitude to be full rank so that matrix $\mat{\Phi}^H\mat{\Phi}$ can be inverted. 

This requirement of $N\geq L$ pilots relates to the fact that the LS estimator gives the same importance to all polynomial coefficients $\beta_l$. In practice though, some, especially higher order ones, may be negligible. This knowledge can help to regularize the problem inversion (even if $N\leq L$). To do this, we leverage the prior information we have about the \gls{pa} transfer function. Indeed, while its exact form is not known, its general behaviour can be generally assumed to be known, \textit{e.g.}, its linear gain or saturation level. A \gls{lmmse} provides an attractive solution, which only requires the first and second-order statistical moments. We now assume that the vector $\vect{\beta}$ is random with mean $\bar{\vect{\beta}}$ and covariance matrix $\mat{C}_{\vect{\beta}}$. In Section~\ref{section:simulation_results}, we discuss two cases of prior information depending if the phase of the channel gain $h$ is known (coherent) or unknown (noncoherent). 
The LMMSE estimator is then
\begin{align*}
	\hat{\vect{\beta}}_{\mathrm{LMMSE}}&=\bar{\vect{\beta}}+(\mat{\Phi}^H\mat{\Phi}+\sigma^2\mat{C}_{\vect{\beta}}^{-1})^{-1}\mat{\Phi}^H(\vect{r}-\mat{\Phi}\bar{\vect{\beta}})
\end{align*}
and the error covariance matrix is
\begin{align*}
	\mat{C}_{\mathrm{LMMSE}}&=\sigma^2(\mat{\Phi}^H\mat{\Phi}+\sigma^2\mat{C}_{\vect{\beta}}^{-1})^{-1}.
\end{align*}
Note that the matrix to invert is always full rank, even if $N \leq L$. Moreover, given the positive definite nature of $\mat{C}_{\vect{\beta}}$ we have $\mat{C}_{\mathrm{LS}}\succeq\mat{C}_{\mathrm{LMMSE}}$, which implies better performance of the LMMSE estimator. Again, we can predict the PA response at $\tilde{N}$ values by $\tilde{\mat{\Phi}}\hat{\vect{\beta}}_{\mathrm{LMMSE}}$ with an error covariance matrix
\begin{align*}
	\tilde{\mat{C}}_{\mathrm{LMMSE}}&=\sigma^2\tilde{\mat{\Phi}}(\mat{\Phi}^H\mat{\Phi}+\sigma^2\mat{C}_{\vect{\beta}}^{-1})^{-1}\tilde{\mat{\Phi}}^H.
\end{align*}

\begin{proposition}\label{prop:all_bases_are_great}
	The basis choice has no impact on the prediction performance related to the LS and LMMSE estimators, \textit{i.e.}, $\tilde{\mat{C}}_{\mathrm{LS}}$ and $\tilde{\mat{C}}_{\mathrm{LMMSE}}$ do not depend on the basis $\mat{\vect{\Psi}}$.
\end{proposition}
\begin{proof}
	We here conduct the proof in the LMMSE case. The LS case can be found by particularization to the case $\mat{C}_{\vect{\beta}}^{-1}=\mat{0}$ (infinite variance of a priori information). We can redo the previous derivation for a general basis $\mat{\Psi}$. We then have $\mat{\Psi}=\mat{\Phi}\mat{U}$ and $\tilde{\mat{\Psi}}=\tilde{\mat{\Phi}}\mat{U}$. Since $\vect{\alpha}=\mat{U}\vect{\beta}$, the random vector $\vect{\alpha}$ has a covariance matrix $\mat{C}_{\vect{\alpha}}=\mat{U}\mat{C}_{\vect{\beta}}\mat{U}^H$. The error covariance matrix then becomes
	\begin{align*}
		\tilde{\mat{C}}_{\mathrm{LMMSE}}&=\sigma^2\tilde{\mat{\Psi}}(\mat{\Psi}^H\mat{\Psi}+\sigma^2\mat{C}_{\vect{\alpha}}^{-1})^{-1}\tilde{\mat{\Psi}}^H\\
		&=\sigma^2\tilde{\mat{\Phi}}\mat{U}(\mat{U}^H\mat{\Phi}^H\mat{\Phi}\mat{U}+\sigma^2\mat{C}_{\vect{\alpha}}^{-1})^{-1}\mat{U}^H\tilde{\mat{\Phi}}^H\\
		&=\sigma^2\tilde{\mat{\Phi}}(\mat{\Phi}^H\mat{\Phi}+\sigma^2\mat{C}_{\vect{\beta}}^{-1})^{-1}\tilde{\mat{\Phi}}^H
	\end{align*}
	where we used the fact that the square matrix $\mat{U}$ is full rank and is thus invertible. This completes the proof.
\end{proof}

\section{Optimal Training Design}\label{subsection:zero_third_order_distortion}

We now consider the optimal design of the pilot symbols $s_n$, $n=0,...,N-1$ or equivalently the matrix $\mat{\Phi}$. As a practical constraint, we consider that the \gls{pa} input power can be limited by a max level. Therefore, we add the constraint $|s_n| \leq 1$, \textit{i.e.}, a unit power constraint which can be generalized straightforwardly to any power $P_{\mathrm{max}}$ by scaling training points $s_n$ by $\sqrt{P_{\mathrm{max}}}$.

In theoretical infinite precision, Prop.~\ref{prop:all_bases_are_great} shows that all bases give same performance. Hence, the basis choice has no impact on the optimal training design minimizing $\tilde{\mat{C}}_{\mathrm{LS}/\mathrm{LMMSE}}$. In practice, when implementing the estimators in finite precision, some bases will perform better and the conventional $\mat{\Phi}$ should be avoided. Indeed, its high order terms $|s|^l$ go to zero fast (as $O(|s|^l)$) as $|s|\rightarrow 0$ while other bases can keep a decay in $O(|s|)$, giving them increased robustness (reduced impact of quantization error)~\cite{raich_orthogonal_2004b}. We now give a lemma that will prove useful for optimal design.

\begin{lemma} \label{lemma:arbitrary_phase}
	The phases of the pilot signals $\angle s_n$ have no impact on the error covariance matrices ${\mat{C}}_{\mathrm{LS}/\mathrm{LMMSE}}$ and $\tilde{\mat{C}}_{\mathrm{LS}/\mathrm{LMMSE}}$. Hence, they can be chosen arbitrarily.
\end{lemma}
\begin{proof}
	The error covariance matrices depend on the matrix
	\begin{align*}
			\mat{\Phi}^H\mat{\Phi}=\begin{pmatrix}
					\sum_n |s_n|^2 & \hdots & \sum_n |s_n|^{L+1}\\
					\vdots & \ddots & \vdots \\
					\sum_n |s_n|^{L+1} & \hdots & \sum_n |s_n|^{2L}
				\end{pmatrix}
		\end{align*}
	which only depends on pilot amplitudes, not phases.
\end{proof}

\subsection{Least Squares Estimator}\label{subsection:training_design_LS}

We now consider the design of $\mat{\Phi}$ to minimize the generalized variance, \textit{i.e.}, the determinant, of the LS estimate error covariance matrix
\begin{align}
	\min_{\substack{s_0,...,s_{N-1} \\0\leq |s_n| \leq 1,\ \forall n}} |\mat{C}_{\mathrm{LS}}|=\sigma^{2L}|(\mat{\Phi}^H\mat{\Phi})^{-1}|. \label{eq:opt_prob}
\end{align}

\begin{theorem}\label{theor:main}
	If $N$ is a multiple of $L$, the optimal training design that solves~(\ref{eq:opt_prob}) divides the $N$ pilots $s_n$ in $L$ sets of $N/L$ pilots. The amplitude of the pilots in each set corresponds to one of the $L$ support points (roots) $t_l$ that solve
	\begin{align}
		(1-t)P_L'(2t-1)=0 \label{eq:root_prob}
	\end{align}
	where $P_L'(t)$ is the derivative of the $L$-th degree Legendre polynomial $P_L(t)$.\footnote{$P_L(2t-1)$ is the shifted Legendre polynomial to the interval $[0,1]$.} The phase of the pilots can be chosen arbitrarily.
\end{theorem}
\begin{proof}
	In light of Lemma~\ref{lemma:arbitrary_phase}, we can simplify the optimization problem by restricting $s_n$ to be fully real and positive. After optimization, a phase can be reinserted if desired, without impacting performance. The fully real problem then has the form of a standard polynomial regression. 
	A large body of literature exists on so-called optimal design of experiments~\cite{pukelsheim2006optimal}. The specific problem of minimizing the determinant is referred to as a D-optimal design in the mathematical statistics community. Two key specificities here are that i) no intercept is considered and ii) pilots should belong to the domain $[0,1]$. These constraints are considered in~\cite[Theor.~1]{huang1995d}, which we tailored to the formalism of this paper.
\end{proof}

Since 1 always is a root of (\ref{eq:root_prob}), it is optimal to set unit magnitude to $N/L$ pilots. For other sets of pilots, roots of $P_L'(2t-1)$ should be found. For orders $L\leq 10$, closed-form expressions can be found easily. 
Indeed, the roots of Legendre polynomial derivatives $P_L'(t)$ come in positive/negative pairs and there is always one in zero for even $L$, \textit{i.e.}, one root in $1/2$ for $P_L'(2t-1)$. In any case, these roots can be pre-computed. 


The optimal design of Theorem~\ref{theor:main} provides an additional global optimality guarantee on the total regression range. To introduce it, let us first consider the PA predicted response for a given input value $\tilde{s}$. Defining $\tilde{\vect{\phi}}(\tilde{s})=(\tilde{s}, ..., \tilde{s}|\tilde{s}|^{L-1}
)^T$, the prediction error variance or \gls{mse} is then given by~(\ref{eq:predicted_C_LS}) particularized for $\tilde{N}=1$
\begin{align*}
	\mathrm{MSE}_{\mathrm{LS}}(\tilde{s})&=\sigma^2\tilde{\vect{\phi}}(\tilde{s})^H(\mat{\Phi}^H\mat{\Phi})^{-1}\tilde{\vect{\phi}}(\tilde{s})
\end{align*}
and the maximal \gls{mse} over the regression range, for a given training matrix $\mat{\Phi}$, is defined as
\begin{align*}
	d_{\mathrm{LS}}(\mat{\Phi})=\max_{\tilde{s}, |\tilde{s}|\leq 1} \mathrm{MSE}_{\mathrm{LS}}(\tilde{s}).
\end{align*}
\begin{corollary}[Minimax MSE] \label{corol:minimax}
	The optimal design of Theor.~\ref{theor:main} also minimizes the maximal prediction MSE over the PA input range, \textit{i.e.}, it is the solution of
	\begin{align*}
		\min_{\substack{s_0,...,s_{N-1} \\0\leq |s_n| \leq 1,\ \forall n}} d_{\mathrm{LS}}(\mat{\Phi})
	\end{align*}
	achieving $d_{\mathrm{LS}}(\mat{\Phi})=\sigma^2L/N$ and thus bounding the prediction error variance as $\mathrm{MSE}_{\mathrm{LS}}(\tilde{s})\leq \sigma^2L/N$.
\end{corollary}
\begin{proof}
	The proof directly comes from the Kiefer-Wolfowitz equivalence theorem~\cite{kiefer_equivalence_1960,pukelsheim2006optimal}, particularized to our case.
\end{proof}

\subsection{Linear Minimum Mean Squared Estimator}

Unfortunately, we do not have the expression of the optimal training matrix in the LMMSE case. We can still put forward certain positive aspects about using the LS optimal design for the LMMSE estimator. In general, given that $\mat{C}_{\mathrm{LS}}\succeq\mat{C}_{\mathrm{LMMSE}}$, we also have that $|\mat{C}_{\mathrm{LS}}|\geq |\mat{C}_{\mathrm{LMMSE}}|$. Hence minimizing $|\mat{C}_{\mathrm{LS}}|$ corresponds to minimizing an upper bound on $|\mat{C}_{\mathrm{LMMSE}}|$, which is generally beneficial. Given the improved performance, we also have the guarantee that the maximal prediction MSE $d_{\mathrm{LMMSE}}(\mat{\Phi})$ is lower or equal than $\sigma^2L/N$. Moreover, given that the matrix $\mat{\Phi}^H\mat{\Phi}$ grows with $N$, it makes sense to normalize it by $1/N$ so that the error covariance matrix becomes
\begin{align*}
	\tilde{\mat{C}}_{\mathrm{LMMSE}}&=\frac{\sigma^2}{N}\tilde{\mat{\Phi}}\left(\frac{1}{N}\mat{\Phi}^H\mat{\Phi}+\frac{\sigma^2}{N}\mat{C}_{\vect{\beta}}^{-1}\right)^{-1}\tilde{\mat{\Phi}}^H.
\end{align*}
It is then clear that, as $\sigma^2/N\rightarrow 0$, the two estimators (and their performance) $\hat{\vect{\beta}}_{\mathrm{LMMSE}}$ and $\hat{\vect{\beta}}_{\mathrm{LS}}$ converge provided that $\mat{\Phi}$ is full rank (thus $N\geq L$). In other words, when the noise variance is small and/or the number of pilots is large, the prior information can be neglected, which intuitively makes sense~\cite{chaloner_bayesian_1995}. This implies that the optimal designs of Section~\ref{subsection:training_design_LS} are asymptotically optimal in the LMMSE case as $\sigma^2/N\rightarrow 0$. 

\section{Simulation Results}\label{section:simulation_results}

In the following, as a benchmark, a uniform pilot allocation is considered, defined as $s_n=1/N+n/N$ for $n=0,...,N-1$.

\subsection{Least Squares Estimator}

\begin{figure}[t!]
	\centering 
	\resizebox{0.95\linewidth}{!}{%
		{\small \input{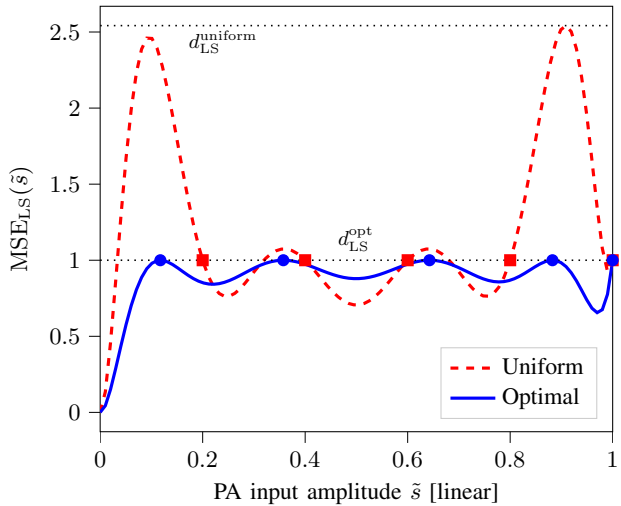}}
	} 
	\vspace{-0.5em}
	\caption{LS prediction (reconstruction) MSE for a uniform and optimal allocation for $L=N=5$ and $\sigma^2=1$. Pilot locations are represented by rectangles and circles respectively.}
	\label{fig:Fig_1_KWT} 
	\vspace{-1em}
\end{figure}

Fig.~\ref{fig:Fig_1_KWT} illustrates the result of Corol.~\ref{corol:minimax} for a fifth polynomial order ($L=5$), $\sigma^2=1$ and $N=L=5$ pilots. The optimal one performs sensibly better than a naive uniform allocation
. Quantitatively, the maximal MSE is more than two times reduced by the optimal allocation.

\begin{figure}[t!]
	\centering 
	\resizebox{0.95\linewidth}{!}{%
		{\small 
\begin{tikzpicture}

\definecolor{darkgray176}{RGB}{176,176,176}
\definecolor{lightgray204}{RGB}{204,204,204}

\begin{axis}[
legend cell align={left},
legend style={
  fill opacity=0.8,
  draw opacity=1,
  text opacity=1,
  at={(0.03,0.97)},
  anchor=north west,
  draw=lightgray204
},
tick align=outside,
tick pos=left,
x grid style={darkgray176},
xlabel={Polynomial order \(\displaystyle L\)},
xmin=1, xmax=8,
xtick style={color=black},
y grid style={darkgray176},
ylabel={Performance gain \(\displaystyle d_{\mathrm{LS}}^{\mathrm{uniform}}/d_{\mathrm{LS}}^{\mathrm{opt}}\)},
ymin=0.0128223518181799, ymax=21.7307306118182,
ytick style={color=black}
]
\addplot [very thick, blue]
table {%
1 1
2 1
3 1.17576319213973
4 1.62956905084786
5 2.54201632831145
6 4.48741855721267
7 9.10700131015958
8 20.7435529636364
};
\addlegendentry{Uniform}
\addplot [very thick, black, dashed]
table {%
1 1
2 1
3 1
4 1
5 1
6 1
7 1
8 1
};
\addlegendentry{No (unit) gain}
\end{axis}

\end{tikzpicture}}
	} 
	\vspace{-0.5em}
	\caption{Performance gain of using the optimal versus uniform allocation in terms of ratio of maximal prediction MSE.}
	\label{fig:Fig_2_d_C_as_function_L} 
	\vspace{-1em}
\end{figure}
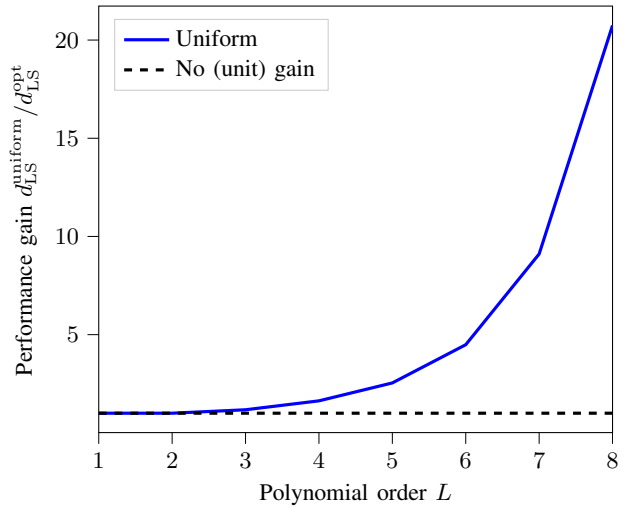

For $L=1$ and $L=2$, the optimal support points are $[1]$ and $[1/2,1]$ and the uniform allocation coincides with the optimal one. As $L$ increases, the performance gap between the uniform allocation and the optimal one gets larger. Fig.~\ref{fig:Fig_2_d_C_as_function_L} illustrates this by plotting the ratio of the maximal prediction MSE $d_{\mathrm{LS}}$ (maximal MSE value in Fig.~\ref{fig:Fig_1_KWT}) for $L=N$. As a reminder, from Corol.~\ref{corol:minimax}, we simply have $d_{\mathrm{LS}}=\sigma^2$ for $N=L$ when using the optimal allocation, \textit{i.e.}, it remains constant. In other words, the alternations of the blue continuous curve in Fig.~\ref{fig:Fig_1_KWT} will never be over $1$, even as $L$ increases. On the other hand, for the uniform one, the symmetric peaks on the left and right of the graphs will continue to increase as $L$ grows. This explains the large improvement observed in Fig.~\ref{fig:Fig_2_d_C_as_function_L}.


\subsection{Linear Minimum Mean Squared Error Estimator}\label{subsection:LMMSE}

The MSE of the LMMSE estimator $\mathrm{MSE}_{\mathrm{LMMSE}}$ does not scale linearly with $\sigma^2$ so that the performance gap with the LS estimator $\mathrm{MSE}_{\mathrm{LS}}$ will generally depend on the SNR which we define as $\mathrm{SNR}={P_{\mathrm{max}}}/{(N\sigma^2)}$. In simulations, we set $P_{\mathrm{max}}=1$ and $N=L$. To determine the a priori statistics of the LMMSE $\bar{\vect{\beta}}$ and $\mat{C}_{\vect{\beta}}$, we consider that the PA follows a random Rapp model without phase distortion 
so that
\begin{align*}
	f(s)&=\frac{Gs}{\left(1+\left(\frac{Gs}{V_{\mathrm{sat}}}\right)^{2S}\right)^{1/2S}}
\end{align*}
where $G\sim \mathcal{N}(1,0.01)$, $V_{\mathrm{sat}}\sim \mathcal{N}(1,0.01)$ and $S\sim \mathcal{N}(2,0.1)$ are the linear gain, the saturation voltage and the smoothness, respectively. We generate $M=100$ realizations of the PA response, for which confidence intervals are shown in Fig.~\ref{fig:Fig_3_Rapp_variations}. We then fit a $L=7$ polynomial order to obtain coefficients $\vect{\beta}_m$ for each realization $m$. As can be seen in the figure, this polynomial approximation is very accurate and cannot be distinguished at first sight from the exact Rapp model response. As discussed in Section~\ref{section:transmission_model}, in practice, the channel $h$, assumed narrowband, will multiply these coefficients. This induces an attenuation $|h|$, which we take into account through the scaling of the noise level $\sigma^2$ but also a phase rotation $e^{\jmath \angle h}$. Depending if this phase is known, \textit{e.g.}, from previous measurements or as part of the calibration process, or not, we consider two cases. i) \textbf{Coherent} case where it is known and can be compensated so that we can estimate the coefficients mean and covariance as $\bar{\vect{\beta}}=1/M\sum_{m=0}^{M-1}\vect{\beta}_m$ and $\mat{C}_{\vect{\beta}}=1/M\sum_{m=0}^{M-1}\vect{\beta}_m\vect{\beta}_m^H-\bar{\vect{\beta}}\bar{\vect{\beta}}^T$. ii) \textbf{Noncoherent} case where phase $\angle h$ is unknown, assumed uniformly distributed in $[0,2\pi]$, causing coefficients to have a zero mean so that $\bar{\vect{\beta}}=\vect{0}$ and $\mat{C}_{\vect{\beta}}=1/M\sum_{m=0}^{M-1}\vect{\beta}_m\vect{\beta}_m^H$.

\begin{figure}[t!]
	\centering 
	\resizebox{0.95\linewidth}{!}{%
		{\small 
\begin{tikzpicture}

\definecolor{darkgray176}{RGB}{176,176,176}
\definecolor{lightgray204}{RGB}{204,204,204}
\definecolor{steelblue31119180}{RGB}{31,119,180}

\begin{axis}[
legend cell align={left},
legend style={
  fill opacity=0.8,
  draw opacity=1,
  text opacity=1,
  at={(0.03,0.97)},
  anchor=north west,
  draw=lightgray204
},
tick align=outside,
tick pos=left,
x grid style={darkgray176},
xlabel={PA input amplitude \(\displaystyle |s|\)},
xmin=0, xmax=1.5,
xtick style={color=black},
y grid style={darkgray176},
ylabel={PA output amplitude \(\displaystyle |f(s)|\)},
ymin=-0.0570135831521859, ymax=1.1972852461959,
ytick style={color=black}
]
\path [draw=blue!20, fill=blue!20]
(axis cs:0,0)
--(axis cs:0,0)
--(axis cs:0.0625,0.050532926388181)
--(axis cs:0.125,0.101081043073624)
--(axis cs:0.1875,0.151675458481041)
--(axis cs:0.25,0.202359019351808)
--(axis cs:0.3125,0.253173437302911)
--(axis cs:0.375,0.304130423091079)
--(axis cs:0.4375,0.355157357566701)
--(axis cs:0.5,0.406016059803783)
--(axis cs:0.5625,0.45621413028087)
--(axis cs:0.625,0.504953442732638)
--(axis cs:0.6875,0.551167726181463)
--(axis cs:0.75,0.593674521273359)
--(axis cs:0.8125,0.631413215299453)
--(axis cs:0.875,0.663686567446418)
--(axis cs:0.9375,0.69030311766854)
--(axis cs:1,0.711561284297259)
--(axis cs:1.0625,0.728103536382177)
--(axis cs:1.125,0.740728107468762)
--(axis cs:1.1875,0.750234901014731)
--(axis cs:1.25,0.757334610226465)
--(axis cs:1.3125,0.762613138179128)
--(axis cs:1.375,0.766530836617802)
--(axis cs:1.4375,0.769438444247499)
--(axis cs:1.5,0.771598153256813)
--(axis cs:1.5,1.14027166304372)
--(axis cs:1.5,1.14027166304372)
--(axis cs:1.4375,1.12807332588177)
--(axis cs:1.375,1.11385801036959)
--(axis cs:1.3125,1.09732651963693)
--(axis cs:1.25,1.07817229595669)
--(axis cs:1.1875,1.05610228304054)
--(axis cs:1.125,1.03086410950919)
--(axis cs:1.0625,1.00227381876888)
--(axis cs:1,0.970231546210171)
--(axis cs:0.9375,0.934707373643537)
--(axis cs:0.875,0.895685745655123)
--(axis cs:0.8125,0.853084055862362)
--(axis cs:0.75,0.806698628472892)
--(axis cs:0.6875,0.756237541005524)
--(axis cs:0.625,0.701447033295909)
--(axis cs:0.5625,0.642268531868977)
--(axis cs:0.5,0.578942061207122)
--(axis cs:0.4375,0.51200705056943)
--(axis cs:0.375,0.442206850262191)
--(axis cs:0.3125,0.370345265322722)
--(axis cs:0.25,0.297153888009676)
--(axis cs:0.1875,0.223208759514548)
--(axis cs:0.125,0.148903700465194)
--(axis cs:0.0625,0.0744665967792082)
--(axis cs:0,0)
--cycle;

\addplot [semithick, black, forget plot]
table {%
0 0
0.0625 0.0624997615836946
0.125 0.124992371769409
0.1875 0.187442108997795
0.25 0.249756453680742
0.3125 0.311759351312817
0.375 0.373168636676635
0.4375 0.433582204068066
0.5 0.492479060505452
0.5625 0.549241331074924
0.625 0.603200238014273
0.6875 0.653702633593493
0.75 0.700186574873126
0.8125 0.742248635580908
0.875 0.77968615655077
0.9375 0.812505245656038
1 0.840896415253715
1.0625 0.865188677575529
1.125 0.885796108488975
1.1875 0.903168592027633
1.25 0.917753453091576
1.3125 0.929969828908027
1.375 0.940194423493694
1.4375 0.948755885064632
1.5 0.955934908150266
};
\addplot [very thick, black]
table {%
0 0
0.0625 0.0624997615836946
0.125 0.124992371769409
0.1875 0.187442108997795
0.25 0.249756453680742
0.3125 0.311759351312817
0.375 0.373168636676635
0.4375 0.433582204068066
0.5 0.492479060505452
0.5625 0.549241331074924
0.625 0.603200238014273
0.6875 0.653702633593493
0.75 0.700186574873126
0.8125 0.742248635580908
0.875 0.77968615655077
0.9375 0.812505245656038
1 0.840896415253715
1.0625 0.865188677575529
1.125 0.885796108488975
1.1875 0.903168592027633
1.25 0.917753453091576
1.3125 0.929969828908027
1.375 0.940194423493694
1.4375 0.948755885064632
1.5 0.955934908150266
};
\addlegendentry{$G=1,V_{\mathrm{sat}}=1,S=2$}
\addplot [very thick, steelblue31119180, mark=x, mark size=3, mark options={solid}, only marks]
table {%
0 0
0.0625 0.0627426128547217
0.125 0.125069703312776
0.1875 0.187320896194187
0.25 0.249566133611198
0.3125 0.311642868736978
0.375 0.373198848359038
0.4375 0.433738778925072
0.5 0.492673170788931
0.5625 0.549367655364456
0.625 0.603191069894869
0.6875 0.653560604545455
0.75 0.69998230652723
0.8125 0.742085235959329
0.875 0.779647568177814
0.9375 0.812612937198621
1 0.841095315042363
1.0625 0.865370721628705
1.125 0.885854059948015
1.1875 0.903059371218019
1.25 0.917541804733172
1.3125 0.92981959711444
1.375 0.940274355667249
1.4375 0.949027940555268
1.5 0.955794240497771
};
\addlegendentry{Polyn. approx. $L=7$ (no noise)}
\addlegendimage{area legend, draw=blue, fill=blue, opacity=0.1}
\addlegendentry{Conf. interval ($\pm 2$ st.dev.)}
\end{axis}

\end{tikzpicture}}
	} 
	\vspace{-0.5em}
	\caption{Random power amplifier response over 100 realizations.}
	\label{fig:Fig_3_Rapp_variations} 
	\vspace{-1em}
\end{figure}
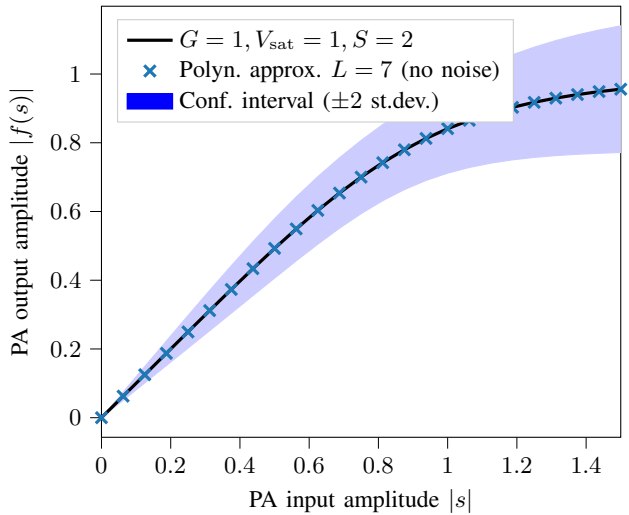

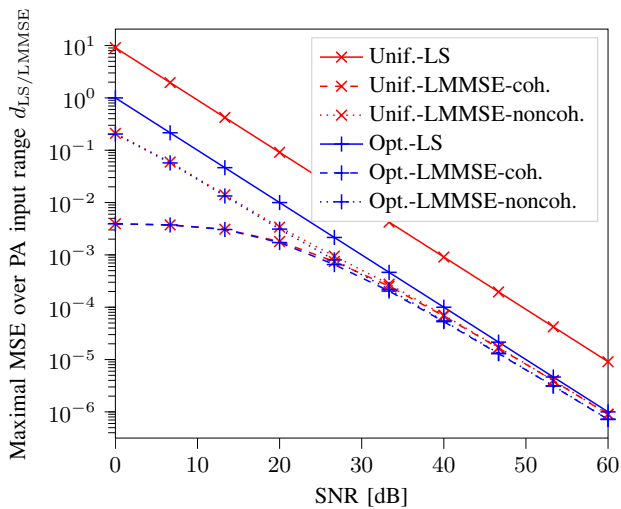
\begin{figure}[t!]
	\centering 
	\resizebox{0.95\linewidth}{!}{%
		{\small 
\begin{tikzpicture}

\definecolor{darkgray176}{RGB}{176,176,176}
\definecolor{lightgray204}{RGB}{204,204,204}

\begin{axis}[
legend cell align={left},
legend style={fill opacity=0.8, draw opacity=1, text opacity=1, draw=lightgray204},
log basis y={10},
tick align=outside,
tick pos=left,
x grid style={darkgray176},
xlabel={SNR [dB]},
xmin=0, xmax=60,
xtick style={color=black},
y grid style={darkgray176},
ylabel={Maximal MSE over PA input range \(\displaystyle d_{\mathrm{LS/LMMSE}}\)},
ymin=3.16063427288857e-07, ymax=20.6345263366562,
ymode=log,
ytick style={color=black}
]
\addplot [semithick, red, mark=x, mark size=3, mark options={solid}]
table {%
0 9.10700136065134
6.66666666666667 1.96204397160182
13.3333333333333 0.422709557033158
20 0.0910700147423995
26.6666666666667 0.0196204397482802
33.3333333333333 0.00422709559185819
40 0.000910700133337272
46.6666666666667 0.000196204393505095
53.3333333333333 4.22709555016887e-05
60 9.10700141851239e-06
};
\addlegendentry{Unif.-LS}
\addplot [semithick, red, dashed, mark=x, mark size=3, mark options={solid}]
table {%
0 0.00392725444723065
6.66666666666667 0.00374329049724678
13.3333333333333 0.00309132979132589
20 0.00180116726976876
26.6666666666667 0.000747317073991463
33.3333333333333 0.000252712933222845
40 6.79308492025606e-05
46.6666666666667 1.66744030604474e-05
53.3333333333333 3.95758081006655e-06
60 8.96807073711477e-07
};
\addlegendentry{Unif.-LMMSE-coh.}
\addplot [semithick, red, dotted, mark=x, mark size=3, mark options={solid}]
table {%
0 0.209837745650506
6.66666666666667 0.0594216032926539
13.3333333333333 0.0140529707206627
20 0.00339062963034421
26.6666666666667 0.000947933763618188
33.3333333333333 0.000276529948635043
40 7.05817353364464e-05
46.6666666666667 1.70425147478319e-05
53.3333333333333 4.0051233401925e-06
60 9.03615625409714e-07
};
\addlegendentry{Unif.-LMMSE-noncoh.}
\addplot [semithick, blue, mark=+, mark size=3, mark options={solid}]
table {%
0 1.00000000554428
6.66666666666667 0.215443471212893
13.3333333333333 0.0464158879937173
20 0.00999999958857245
26.6666666666667 0.00215443472117194
33.3333333333333 0.00046415887551432
40 9.99999997256529e-05
46.6666666666667 2.15443470703361e-05
53.3333333333333 4.64158878698412e-06
60 9.99999997244958e-07
};
\addlegendentry{Opt.-LS}
\addplot [semithick, blue, dashed, mark=+, mark size=3, mark options={solid}]
table {%
0 0.00392540099827274
6.66666666666667 0.00373502798236322
13.3333333333333 0.00305917600975963
20 0.00172056678275776
26.6666666666667 0.000653562721748491
33.3333333333333 0.000203561498409615
40 5.30011383333686e-05
46.6666666666667 1.29328609843224e-05
53.3333333333333 3.10716711261623e-06
60 7.16132440983769e-07
};
\addlegendentry{Opt.-LMMSE-coh.}
\addplot [semithick, blue, dotted, mark=+, mark size=3, mark options={solid}]
table {%
0 0.204067264972424
6.66666666666667 0.0571897823469335
13.3333333333333 0.0133634035977825
20 0.00310879978591108
26.6666666666667 0.000802447302517828
33.3333333333333 0.00021923228082626
40 5.47704522059853e-05
46.6666666666667 1.32144586732076e-05
53.3333333333333 3.15135305336708e-06
60 7.23803554371995e-07
};
\addlegendentry{Opt.-LMMSE-noncoh.}
\end{axis}

\end{tikzpicture}}
	} 
	\vspace{-0.5em}
	\caption{Maximal prediction variance for LS/LMMSE estimators and uniform versus optimized training.}
	\label{fig:Fig_4_LS_LMMSE_comparison} 
	\vspace{-1em}
\end{figure}

A performance comparison is given in Fig.~\ref{fig:Fig_4_LS_LMMSE_comparison} showing the maximal prediction variance $d_{\mathrm{LS/LMMSE}}(\mat{\Phi})$ (for the LS or LMMSE estimator) as a function of the SNR. A few important observations can be made. i) Using the optimized allocation of Theor.~\ref{theor:main} provides most gain as compared to the uniform one when using the LS estimator (factor 10 in accordance to Fig.~\ref{fig:Fig_2_d_C_as_function_L}) while the improvement is negligible when using the LMMSE estimator. ii) At low SNR, the LMMSE estimator has a large gain as compared to the LS estimator given the prior information, especially in the coherent case where even more prior information is available: factor 300 (coherent) versus 5 (noncoherent) at 0 dB SNR. At high SNR, their performance converges to the LS one, as expected from Section~\ref{subsection:LMMSE}. iii) Interestingly, even at 60 dB SNR, when using the uniform allocation, the LMMSE estimator has not yet converged to the LS one. It achieves a similar performance as the LS/LMMSE one with optimized training. This means that, when using the LMMSE estimator, a simpler uniform allocation would provide close-to-optimal solution and can be considered sufficient.

%

\section{Conclusion}\label{section:conclusion}

In this work, we considered OTA PA calibration using estimation theory. The LS and LMMSE estimators were rederived. The optimal training in the LS case was obtained by minimizing the generalized variance of the error covariance matrix. This optimal design also minimizes the maximal MSE when reconstructing the PA over its full input range. Via simulations, we show an improvement of a factor 10 for a $L=7$ PA order. Leveraging prior information, the LMMSE estimator provides an additional gain of a factor 5 and 300 in the noncoherent and coherent cases respectively.

 \section*{Acknowledgment}
This research was partially funded by 6GTandem, supported by the Smart Networks and Services Joint Undertaking (SNS JU) under the European Union’s Horizon Europe research and innovation programme under Grant Agreement No 101096302.

\scriptsize 
\bibliographystyle{IEEEtran}
\bibliography{IEEEabrv,IEEEreferences}

\begin{thebibliography}{10}
\providecommand{\url}[1]{#1}
\csname url@samestyle\endcsname
\providecommand{\newblock}{\relax}
\providecommand{\bibinfo}[2]{#2}
\providecommand{\BIBentrySTDinterwordspacing}{\spaceskip=0pt\relax}
\providecommand{\BIBentryALTinterwordstretchfactor}{4}
\providecommand{\BIBentryALTinterwordspacing}{\spaceskip=\fontdimen2\font plus
\BIBentryALTinterwordstretchfactor\fontdimen3\font minus \fontdimen4\font\relax}
\providecommand{\BIBforeignlanguage}[2]{{%
\expandafter\ifx\csname l@#1\endcsname\relax
\typeout{** WARNING: IEEEtran.bst: No hyphenation pattern has been}%
\typeout{** loaded for the language `#1'. Using the pattern for}%
\typeout{** the default language instead.}%
\else
\language=\csname l@#1\endcsname
\fi
#2}}
\providecommand{\BIBdecl}{\relax}
\BIBdecl

\bibitem{auer11}
G.~Auer \emph{et~al.}, ``How much energy is needed to run a wireless network?'' \emph{{IEEE} Wireless Commun. Mag.}, vol.~18, no.~5, pp. 40--49, 2011.

\bibitem{lavr10}
P.~M. Lavrador \emph{et~al.}, ``The linearity-efficiency compromise,'' \emph{{IEEE} Microw. Mag.}, vol.~11, no.~5, pp. 44--58, 2010.

\bibitem{Fager2019}
C.~Fager \emph{et~al.}, ``{Linearity and Efficiency in 5G Transmitters: New Techniques for Analyzing Efficiency, Linearity, and Linearization in a 5G Active Antenna Transmitter Context},'' \emph{{IEEE} Microw. Mag.}, vol.~20, no.~5, pp. 35--49, 2019.

\bibitem{cripps2006rf}
S.~C. Cripps, \emph{{RF power amplifiers for wireless communications}}.\hskip 1em plus 0.5em minus 0.4em\relax Artech House, 2006, vol.~2.

\bibitem{Z3RO_journal}
F.~Rottenberg, G.~Callebaut, and L.~Van~der Perre, ``{The Z3RO Family of Precoders Cancelling Nonlinear Power Amplification Distortion in Large Array Systems},'' \emph{{IEEE} Trans. Wireless Commun.}, vol.~22, no.~3, pp. 2036--2047, 2023.

\bibitem{10279810}
T.~Feys, X.~Mestre, and F.~Rottenberg, ``Self-supervised learning of linear precoders under non-linear pa distortion for energy-efficient massive mimo systems,'' in \emph{ICC 2023}, 2023, pp. 6367--6372.

\bibitem{huang_signal_2014}
X.~Huang, Z.~Zhu, and H.~Leung, \emph{\BIBforeignlanguage{en}{Signal processing for {RF} circuit impairment mitigation}}.\hskip 1em plus 0.5em minus 0.4em\relax Boston: Artech house, 2014.

\bibitem{9780639}
A.~Ben~Ayed \emph{et~al.}, ``Digital predistortion of millimeter-wave arrays using near-field based transmitter observation receivers,'' \emph{{IEEE} Trans. Microw. Theory Tech.}, vol.~70, no.~7, pp. 3713--3723, 2022.

\bibitem{8877285}
N.~Tervo \emph{et~al.}, ``Digital predistortion concepts for linearization of mmw phased array transmitters,'' in \emph{2019 16th International Symposium on Wireless Communication Systems (ISWCS)}, 2019, pp. 325--329.

\bibitem{tervo_thesis}
N.~Tervo, ``Concepts for radiated nonlinear distortion and spatial linearization in millimeter-wave phased arrays.'' PhD thesis, University of Oulu, March 2022.

\bibitem{ding_effects_2004}
L.~Ding and G.~Zhou, ``\BIBforeignlanguage{en}{Effects of {Even}-{Order} {Nonlinear} {Terms} on {Power} {Amplifier} {Modeling} and {Predistortion} {Linearization}},'' \emph{\BIBforeignlanguage{en}{{IEEE} Trans. Veh. Technol.}}, vol.~53, no.~1, pp. 156--162, Jan. 2004.

\bibitem{raich_orthogonal_2004b}
R.~Raich, H.~Qian, and G.~Zhou, ``\BIBforeignlanguage{en}{Orthogonal {Polynomials} for {Power} {Amplifier} {Modeling} and {Predistorter} {Design}},'' \emph{\BIBforeignlanguage{en}{{IEEE} Trans. Veh. Technol.}}, vol.~53, no.~5, pp. 1468--1479, Sep. 2004.

\bibitem{kay_fundamentals_1993}
S.~M. Kay, \emph{Fundamentals of statistical signal processing}, ser. Prentice {Hall} signal processing series.\hskip 1em plus 0.5em minus 0.4em\relax Prentice-Hall PTR, 1993.

\bibitem{pukelsheim2006optimal}
F.~Pukelsheim, \emph{{Optimal Design of Experiments}}.\hskip 1em plus 0.5em minus 0.4em\relax Society for Industrial and Applied Mathematics, 2006.

\bibitem{huang1995d}
M.-N.~L. Huang, F.-C. Chang, and W.~K. Wong, ``{D-optimal designs for polynomial regression without an intercept},'' \emph{Statistica Sinica}, pp. 441--458, 1995.

\bibitem{kiefer_equivalence_1960}
J.~Kiefer and J.~Wolfowitz, ``\BIBforeignlanguage{en}{The {Equivalence} of {Two} {Extremum} {Problems}},'' \emph{\BIBforeignlanguage{en}{Canadian Journal of Mathematics}}, vol.~12, pp. 363--366, 1960.

\bibitem{chaloner_bayesian_1995}
K.~Chaloner and I.~Verdinelli, ``\BIBforeignlanguage{en}{Bayesian {Experimental} {Design}: {A} {Review}},'' \emph{\BIBforeignlanguage{en}{Statistical Science}}, vol.~10, no.~3, Aug. 1995.

\end{thebibliography}

\end{document}